%% file: cutsize.tex
\documentclass[draftclsnofoot,12pt,onecolumn] {IEEEtran}

\usepackage{mdwlist}

\usepackage{flushend}
\usepackage{url}
\usepackage{graphicx}
\usepackage{subfigure}
\usepackage{xcolor,colortbl}
\usepackage[mathscr]{eucal}
\usepackage{amsthm }
\usepackage{amssymb}
\usepackage{cases}

\usepackage{mathptmx}

\usepackage{wrapfig}

\usepackage{subfigure}
\usepackage{algorithm}
\usepackage[noend]{algorithmic}

\floatname{algorithm}{Algorithm}

\usepackage{mathptmx}
\usepackage{mdwlist}
\usepackage{flushend}
\usepackage[mathscr]{eucal}
\usepackage{xcolor,colortbl}
\usepackage{url}
\usepackage[mathscr]{eucal}
\usepackage{amsthm }
\usepackage{amssymb}
\usepackage{mdwlist}
\usepackage{url}
\usepackage{graphicx}
\usepackage{subfigure}

\newtheorem{definition}{Definition}[section]
\newtheorem{lemma}{Lemma}[section]
\newtheorem{theorem}{Theorem}[section]

\newtheorem{corollary}{Corollary}[section]

\usepackage{mathptmx}
\usepackage{flushend}
\usepackage{xcolor,colortbl}
\definecolor{mygray11}{gray}{.99}
\definecolor{mygray10}{gray}{.89}
\definecolor{mygray9}{gray}{.79}
\definecolor{mygray8}{gray}{.69}
\definecolor{mygray7}{gray}{.59}
\definecolor{mygray6}{gray}{.49}
\definecolor{mygray5}{gray}{.39}
\definecolor{mygray4}{gray}{.29}
\definecolor{mygray3}{gray}{.19}
\definecolor{mygray2}{gray}{.1}
\definecolor{mygray}{gray}{.01}

\definecolor{mygray7test}{gray}{.973}
\definecolor{mygray6test}{gray}{.906}
\definecolor{mygray5test}{gray}{.7532}
\definecolor{mygray4test}{gray}{.5998}
\definecolor{mygray3test}{gray}{.42}
\definecolor{mygray2test}{gray}{.25}
\definecolor{mygraytest}{gray}{.01}

\ifCLASSOPTIONcompsoc
\else
\fi
\ifCLASSINFOpdf
\else
\fi
\begin{document}

%
\title{\huge A Simple Algorithm for Computing BOCP  }
%
%
%
%

\author{Jack~Wang 
\IEEEcompsocitemizethanks{\IEEEcompsocthanksitem Contact information: cszjwang@gmail.com. 
\protect\\

}
\thanks{}}

\IEEEcompsoctitleabstractindextext{%
\begin{abstract}
In this article,   we  devise a concise  algorithm for computing BOCP. Our method is  simple, easy-to-implement but without loss of efficiency.  Given two circular-arc polygons with $m$ and $n$ edges respectively, our method runs in $O(m+n+(l+k)\log l)$ time, using $O(m+n+k)$ space, where  $k$ is the number of intersections, and $l$ is the number of {edge}s. Our algorithm has the power to approximate to linear complexity  when $k$ and $l$ are small.   The superiority of the proposed algorithm is also validated through  empirical study.   
\end{abstract}

}

\maketitle

\IEEEdisplaynotcompsoctitleabstractindextext

%
\IEEEpeerreviewmaketitle





\input{mainText}


{ \normalsize
\bibliographystyle{abbrv}
\bibliography{sample}
}


\end{document}

%% file: mainText.tex
\section{Introduction}\label{sec:1}
{B}oolean operation on polygons is one of the oldest and best-known problems in computer graphics, and it has attracted much   attentions, due to its simple formulation and broad applications in  various disciplines  such as computational geometry, CAD, GIS, visual computing, motion planning  \cite{MarkdeBerg:computational,JamesDFoley:computer,YvonGardan:anAlgorithm,JohnEHopcroft:robust,Francisco:aNew,KevinWeiler:HiddenSurface,ZhijieWang:prqumo,ZhijieWang:prqumoUsingNothing}.   
When the polygons to be operated are  \textit{conic polygons} (whose boundaries consist of  conic segments or second degree curves), researchers have  made some efforts, see, e.g.,   \cite{EricBerberich:aComputational,EricBerberich:aGeneric,Yong-XiGong:Boolean}.  
The conic polygon has several  special or degenerate cases: (\romannumeral 1)  the linear polygon (known as traditional polygon), whose boundaries consist of only linear curves, i.e., straight line segments; and (\romannumeral 2) the circular-arc polygon, whose boundaries consist of circular-arcs and/or straight line segments.  
The  natural problem --- boolean operation on traditional polygons ---  has been extensively investigated,  see e.g., \cite{RDAndreev:Algorithm,GuntherGreiner:efficient,DKrishnan:systolic,You-DongLiang:AnAnalysis,YongKuiLiu:AnAlgorithm,Patrick-GillesMaillot:aNew,AvrahamMargalit:Analgorithm,Francisco:aNew,YuPeng:aNew,AriRappoport:AnEfficient,MRivero:Boo,BalaRVatti:aGeneric}. However, in existing literature (almost) no  article  \textit{focuses  on} another natural problem --- boolean operation on circular-arc polygons.   In fact, boolean operation  on  circular-arc polygons also has   many applications.  
For instance, deploying sensors to ensure wireless coverage is an important problem \cite{SeapahnMeguerdichian:coverage,bangWang:coverage}. The sensing range of a single sensor is  a circle. With polygonal obstacles, its sensing range is cut off, shaping  a circular-arc polygon.  When we  verify the wireless coverage range of  sensors, boolean operation on  circular-arc polygons is needed.  As another example, assume there are a group of free-rotating cameras used to  monitor a supermarket. The visual range of a single camera can be regarded as a circle, as it is to be freely rotated.  Various obstacles such as goods shelves usually impede the visions of  cameras,  here boolean operation can be used to check the blind angles. Last but not least, consider  a group of free-moving robots used to guide the visitors in a museum. Since the energy of a single robot is limited, its movable region is restricted to a circle. With the impact of various obstacles such as exhibits, the original movable region is to be cut off by these obstacles. When we verify if every place in the museum can be served by at least a robot,  boolean operation on  circular-arc polygons is also needed.

Although the solution used to handle conic (or more general) polygons can  also work for  its special cases, the special cases, however, usually have their unique properties;  directly executing the  algorithm used to handle  conic (or more general) polygons is usually not efficient enough. It is just like  when the applications only involve  traditional polygons, we usually incline to use the solutions targeted for traditional polygons rather than  the ones for  conic (or more general) polygons. With the similar argument, when the applications only involve  circular-arc polygons, a targeted solution  for  circular-arc polygons should be favourable for potential users. 

Motivated by these, this paper makes the effort  to the problem of \textit{boolean operation on circular-arc polygons}. In particular, we are interested in developing   algorithms with the following features: (\romannumeral 1)   easy-to-implement for deployment in practice, and (ii)   having nice theoretical guarantees. To this end, first of all  a concise and easy-to-operate data structure is naturally developed (Section \ref{subsec:datastructure});  based on this concise  structure, we 
then propose an  algorithm dubbed as R{\footnotesize E2L} that consists of three main steps. 

The first step  is the kernel (or core) of R{\footnotesize E2L}, yielding two \textit{special} sequence lists. Specifically, the kernel  integrates three    simple yet efficient  strategies:  
(\romannumeral 1) it introduces the concept of {\underline{\textbf{r}}elated \underline{\textbf{e}}dges}, which is used  to avoid irrelevant computation as much as possible; (\romannumeral 2) it employs two \textit{special} sequence \underline{\textbf{l}}ists, each  one is a compound structure with three domains; they are used  to let the  \textit{decomposed arcs},  intersections and \textit{processed related edges} be well organized, and  thus immensely simplify the subsequent computation; and  (\romannumeral 3)  it assigns {two \underline{\textbf{l}}abels}  to each processed  related edge  before the edge is placed into a  balanced tree; this  contributes to  avoiding the ``false'' intersections being reported, and speeding up the process of inserting the reported intersections into their corresponding edges (Section \ref{sec:our solution}). 
The second  step  produces two \textit{new} linked lists in which  the intersections, appendix points, and original vertices have been arranged, and the decomposed arcs have been merged. To obtain these two new linked lists, two important but easy-to-ignore issues, ``inserting new appendix points'' and ``merging the decomposed arcs'', are addressed (Section \ref{sec:two new linked lists}). 
The third step is to obtain the resultant (or output) polygon by traversing these two new linked lists. In order to correctly traverse them,  the \textit{entry-exit} properties are naturally adopted, and three  traversing rules are developed   (Section \ref{sec:traversing}).

Viewed from  a macro perspective, similar to many methods (see e.g., \cite{Sutherland:reentrant,You-DongLiang:AnAnalysis,RDAndreev:Algorithm, BalaRVatti:aGeneric,GuntherGreiner:efficient,YongKuiLiu:AnAlgorithm}) in the literature, our solution  also partially  inherits two  well-known proposals:  Bentley-Ottmann Plane Sweep algorithm \cite{JonLouisBentley:Algorithms}  and   Weiler-Atherton Clipping algorithm \cite{KevinWeiler:HiddenSurface}, whereas we also  advance  existing results from various aspects. To summarize,    we make the following main  contributions. 


\begin{enumerate*}
\item We highlight  the circular-arc polygon is one of special cases of the conic polygon, and boolean operation on circular-arc polygons also has many applications. 
\item We devise a concise and easy-to-operate data structure, and develop a targeted algorithm for boolean operations on circular-arc polygons.  
\item While this paper focuses on boolean operations of circular-arc polygons, we show our techniques can be easily extended to compute boolean operations of other types of polygons (Section \ref{sec:extension}).
\item We provide the rigorous and detailed theoretical analysis for our  algorithm. In brief, given two circular-arc polygons with $m$ and $n$ edges respectively, our algorithm runs in $O(m+n+(l+k)\log l)$ time, using $O(m+n+k)$ space, where  $k$ is the number of intersections, and $l$ is the number of {related edge}s  (Section \ref{sec:complexity}).
\item We conduct extensive experiments to demonstrate the efficiency and effectiveness of our solution  (Section \ref{sec:experiment}). 
\end{enumerate*}

The novelty of our work is threefold:  to  the best our knowledge  (\romannumeral 1) it is the first comprehensive study on boolean operations of  circular-arc polygons; (\romannumeral 2) it is the first time to employ the idea    ``utilizing related edges'' for boolean operations of polygons, this technique   is simple enough to be practical value; and (\romannumeral 3)  it is the first output-sensitive algorithm having the potential to approximate to linear complexity for   boolean operations of polygons.  

Next, we review previous works most related to ours, and then present  our  algorithm including rigorous theoretical analysis and extensive empirical study.

\section{Related Work}\label{sec:related work}

We first clarify several technical terms for ease of presentation. It is well known that there are three typical boolean operations: intersection, union, and difference. Note that \textit{polygon clipping} mentioned in many papers  is actually to compute the \textit{difference} of two polygons \cite{BalaRVatti:aGeneric}. Given two polygons, the one  to be clipped is called the \textit{subject polygon},  another  is usually called the \textit{clip polygon} or \textit{clip window} \cite{You-DongLiang:AnAnalysis,Sutherland:reentrant,BalaRVatti:aGeneric}. Given a polygon, if there is a pair of non-adjacent edges  intersecting with each other,   this polygon is usually called the \textit{self-intersection} polygon \cite{RDAndreev:Algorithm,GuntherGreiner:efficient,MRivero:Boo}. Throughout this paper, the \textit{traditional polygon} refers to the polygon whose boundaries consist of \textit{only} straight line segments, while the \textit{circular-arc polygon} refers to the polygons whose  boundaries  consist of  circular arcs, or  both  straight line segments and circular arcs. 
We are now ready to  review the previous works most related to ours.

\subsection{Boolean Operations on Traditional Polygons} 
\label{subsec:booleanOpeTrandition}
In existing literature, there are many papers studying  boolean operation of traditional polygons. For example,  Sutherland-Hodgeman  \cite{Sutherland:reentrant}  proposed an elegant  algorithm dealing with the case when the \textit{clip polygon}  is convex. Liang et al.  \cite{You-DongLiang:AnAnalysis}  gave elaborate analysis on the case when the \textit{clip polygon} is rectangular. Andreev \cite{RDAndreev:Algorithm} presented an algorithm dealing with the case when  the \textit{subject polygon} is with  holes and self-intersections. Vatti \cite{BalaRVatti:aGeneric}  and Greiner-Hormann \cite{GuntherGreiner:efficient} proposed general algorithms that can  handle  concave polygons with holes and self-intersections for both the clip and the subject polygons.  Later, Liu et al. \cite{YongKuiLiu:AnAlgorithm} further optimized  Greiner-Hormann's  algorithm. Rivero-Feito  \cite{MRivero:Boo}   achieved  boolean operation of polygons based on  the concept of  \textit{simplicial chains}. Peng et al. \cite{YuPeng:aNew} also adopted {simplicial chains}, and  improved Rivero-Feito's algorithm. Recently,  Martinez et al. \cite{Francisco:aNew} proposed to  subdivide the edges at the intersection. These works lay a solid foundation for the future research. Compared to these works, this paper  focuses on boolean operation of circular-arc polygons, and thus is  different from theirs.

\subsection{Boolean Operations on Conic/General Polygons} 
\label{subsec:booleanOnGeneralPolygon}
Researchers have also made some efforts on boolean operations of conic polygons. For example, 
Berberich et al. \cite{EricBerberich:aComputational} proposed to decompose non-x-monotone curves and compute the arrangement of segments using the plane sweep method, and then   compute the \textit{overlap} of two polygons using the results of arrangement, in order to achieve boolean operations. Gong et al. \cite{Yong-XiGong:Boolean} achieved boolean operation of conic polygons using the topological relation between two conic polygons, this method does not require \textit{x-monotone} conic arc segments. Both  algorithms can support boolean operation of circular-arc polygons, as the conic polygon is the general case of the circular-arc polygon. 
Moreover, the computational geometry algorithms library   (C{\footnotesize GAL}) \cite{fogel:cgalArrangements}  can also support boolean operation of circular-arc polygons. Inspecting the source codes of C{\footnotesize GAL}, we realize that its idea is  directly invoking   the algorithm of boolean operation on \textit{general polygons},  defined as \texttt{General} \texttt{Polygon\_2} in C{\footnotesize GAL}{\footnotesize \footnote{\footnotesize More information please refer to the site: \url{http://www.cgal.org}}}. To some extent the general polygon can be considered as the most general case, as  its edges  can be line segments, circular arcs, conic curves, cubic curves, or even more complicated curves. Although  the essence of the algorithm in C{\footnotesize GAL} is basically similar to that in \cite{EricBerberich:aComputational} (using  the \textit{plane sweep method} to compute the intersections, and the D{\footnotesize CEL}  structure to represent the  polygons),      C{\footnotesize GAL} is a very powerful and  useful library collecting  many classical  ideas. For example,  Emiris et al. \cite{IoannisZEmiris:towardsAndOpen} developed a kernel, for curved objects and related operations, that was targeted for inclusion in  C{\footnotesize GAL}.  Note that, the C{\footnotesize GAL} project itself also yields many  nice papers  in which boolean operation on polygons with  curves is    mentioned, see e.g., \cite{EricBerberich:aGeneric,RonWein:Exact}, to mention just a few.   These excellent works  are the cornerstones of our study, giving us a lot of inspiration.

Compared to these works, our work is different from theirs in the following  aspects at least.  
First, this paper focuses on one of special cases of conic polygon;   specially, we give insights into its unique properties, design a concise data structure customized for  this special case, and develop a targeted algorithm, in which the central idea  `utilizing related edges' (accompanied with a set of well-designed strategies)   is proposed.  To our knowledge, it is the first time to employ this technique for  boolean operations on polygons. Moreover, we give the rigorous theoretical analysis for our algorithm, which runs in  $O( m+n+(l+k)\log l  )$ time, and  approximates to linear complexity when $k$ and $l$ are small (notice: the best   known result for polygon boolean operation   runs in $O( (m+n+k)\log (m+n))$ time, which is no better than \textit{linearithmic
time}{\small \footnote{Simply speaking, linearithmic
time in Big O notation refers to $O(N\log N)$, provided that the input  is $O(N)$ size. }} even if $k$  is small); its superiority is also verified by extensive experiments. 


\section{Preliminary} \label{sec:warmup}
\vspace{-1ex}

\subsection{Data Structure}\label{subsec:datastructure}

\vspace{-1ex}
It is well known that the  traditional polygon can be represented by a series of vertices. This method however  is invalid for   polygons containing  circular arcs, as  two vertices  cannot exactly determine  a circular arc segment (note: it may be a major or minor arc). Even so, this ambiguity can be easily eliminated by adding  an \textit{appendix point}, where the appendix point can be an arbitrary point that is located on the  arc but it is not the endpoints of the arc. 
For clarity,  a traditional vertex  is denoted by $v_i$, and an appendix point  is denoted by $\widetilde{v_j}$. For example, $\{v_1,\widetilde{v_2},v_3,v_4,v_5\}$ determine a circular-arc polygon with four edges (including one circular arc segment $\widehat{v_1\widetilde{v_2} v_3}$ and three straight line segments $\overline{v_3v_4}$, $\overline{v_4v_5}$, $\overline{v_5v_1}$).  Unless stated otherwise, in the rest of the paper we always use $\overline{~\cdot~}$ and $\widehat{~~\cdot~~}$ to denote the line segment and the arc segment, respectively.  
In order to efficiently operate   circular-arc polygons, we devise a  data structure called A{\footnotesize  PDLL} (appendix point based doubly linked list). Specifically, each node in the   list consists of several domains below.
\vspace{-1ex}

\begin{itemize*}
\item \textsf{\small Data:} $(x,y)$,  the coordinates of a  point.
\item \textsf{\small Tag:} Boolean type, it  indicates whether this point is a traditional vertex or an appendix point.
\item \textsf{\small Crossing:} Boolean type, it indicates whether this point is an intersection.
\item \textsf{\small EE:} Boolean type, it indicates what property (\textit{entry} or \textit{exit}) an intersection  has.  
\item \textsf{\small Prev:} Node pointer, it points to the previous node.
\item \textsf{\small Next:} Node pointer, it points to the next node.
\end{itemize*}

\vspace{-2ex}


\subsection{Observation}\label{subsec:observation}
In this subsection, we  introduce a simple yet important observation that  will be  frequently used  later.    To explain, we need some preliminaries.  

\vspace{-1ex}

\begin{definition}
{({\upshape Non-x-monotone} circular arc)}
Given any circular arc, it is a non-x-monotone circular arc such that there is at least one vertical line that intersects with the circular arc at two points. 
\end{definition}

\vspace{-2ex}

\begin{definition}
{ ({\upshape X-monotone} circular arc)} 
A circular arc is an x-monotone circular arc   such  that there is at most  one intersection  with any vertical line.
\end{definition}

\vspace{-1ex}

Lemma \ref{lemma:at least at most} below  formalizes our observation, which can be viewed as a unique property of circular-arc polygons (compared to other types of polygons). 

\vspace{-1ex}

\begin{lemma}
Let $N_{xmc}$ be an arbitrary  non-x-monotone circular arc,  and $C$ be its corresponding circle. Assume that $l_h$ is  a horizontal line  passing through  the center of $C$. We have that $l_h$ can decompose  $N_{xmc}$ into  at least two and at most three  x-monotone arcs.
\label{lemma:at least at most}
\end{lemma}

\vspace{-2ex}

\begin{proof}
It is immediate by  \textit{analytic geometry}.  
\end{proof}


\vspace{-4ex}

\section{The Kernel of RE2L}  \label{sec:our solution}
In this section, we detail Step 1  of our solution. Specifically, we first expatiate  the main ideas integrated in Step 1 (Sections \ref{subsec:relatedEdge}-\ref{subsec:avoidingFalse}), and then  present the detailed algorithms (Section \ref{subsec:Step1Algorithm}).

\vspace{-2ex}
\subsection{Utilizing Related Edges}\label{subsec:relatedEdge}

One of our strategies is to choose \textit{related edges} (defined later) before doing others. The  purpose of choosing \textit{related edges} is to avoid   operations that are irrelevant with obtaining the final result as much as possible. To  define related edges formally, we need two notions.

\begin{definition}
{\upshape (Extended boundary lines)} 
Given a circular-arc  polygon, w.l.o.g. (without loss of generality), assume the coordinates of left-bottom  corner of its  MBR  (minimum bounding rectangle) are ($x_1$,$y_1$),  the one of right-top corner are  ($x_2$,$y_2$). Then, the following four lines, X=$x_1$, X=$x_2$, Y=$y_1$, Y=$y_2$ are respectively the left, right, bottom and top extended boundary lines of this circular-arc polygon. 
\end{definition}

\begin{definition}
{\upshape (Effective axis)} 
Let  $I_{mm}$ be the intersection set of two circular-arc polygons' MBRs. If the horizontal span of $I_{mm}$ is larger or equal to its vertical span; then, the y-axis is the effective axis. Otherwise, the x-axis is the effective axis.
\end{definition}

We now provide the formal definition and   inspect more properties of related edges.  

\begin{definition}
\label{def:related edges}
{\upshape (Related edges)}
Let $l_1$($l_2$) and $r_1$($r_2$) be the left and right  extended boundary lines of the circular-arc polygon $\mathscr P _1$($\mathscr P _2$), respectively. W.l.o.g.,  assume the effective axis is the x-axis and $l_1< l_2<r_1< r_2$, where $l_1<l_2$ denotes  $l_1$ is in the left of $l_2$. Then, the following edges are related edges:    (\romannumeral 1) edges  located between   $l_2$ and $r_1$; or (2)  edges  intersected with $l_2$ or $r_1$. 
\end{definition}

See Figure \ref{fig:many edges a} for an example, edges  $\overline{ab}$ and $\overline{bc}$ are \textit{related edges} as they intersect with $l_2$. Similarly, edges $\overline{de}$ and $\overline{ef}$  are also \textit{related edges}. We remark that in Definition \ref{def:related edges}  there  are actually other cases, e.g., ``$l_1< l_2<r_2< r_1$'' or the effective axis is the y-axis;  these cases are similar to the listed case,  omitted.

\begin{figure}[h]
  \centering
  \subfigure[\scriptsize { } ]{\label{fig:many edges a}
     \includegraphics[scale=.5]{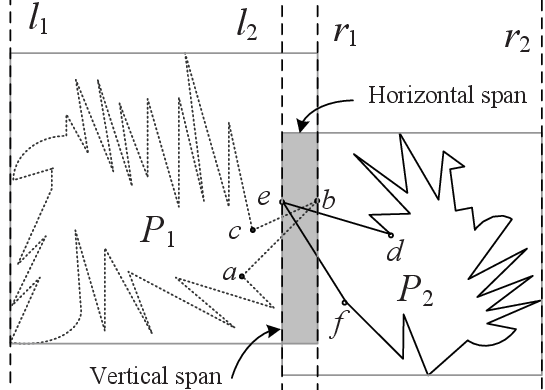}} 
    \hspace{0ex}
  \subfigure[\scriptsize { }]{\label{fig:many edges b}
      \includegraphics[scale=1]{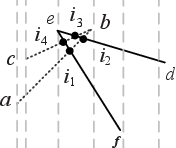}} 
      \vspace{-2ex} 
\caption{\small  Example of \textit{related edges}. (a) Two big rectangles denote the MBRs; the grey rectangle denotes the intersection set of two MBRs, and the dashed vertical lines  denote the extended boundary lines. (b) Partial enlarged drawing. } 
 \label{fig:many edges}
\end{figure}

\begin{definition}\label{definition:processedRelatedEdges}
{\upshape (Processed related edges)}
Given a number  of  related edges, we decompose them if there are non-x-monotone arcs, we call all the edges (after decomposing)  the processed related edges.
\end{definition}

By Lemma \ref{lemma:at least at most} and Definition \ref{definition:processedRelatedEdges}, we have the following corollary (which will be used later).
\begin{corollary}\label{lemma:processed related edges}
Given  $l$  related edges, if there is no non-x-monotone arc among them, the number of processed related edges is $l$. Otherwise, the number of processed related edges is larger than $l$ and  no more than $3l$.
\end{corollary}

Up to now, we have discussed the properties of related edges, and briefly explained how to choose related edges from two circular-arc polygons (remark: more explanations will be given in Algorithm 1 and in the proof of Lemma \ref{theorem:modified plane sweep algorithm complexity}). We next show how to use two sequence lists to manage the 
\textit{processed related edges} and other important components.


\vspace{-1ex}
\subsection{Managing Important Components}\label{subsec:two sequence lists}
The main purpose of the two sequence lists (i.e., arrays) is to let the processed related edges, intersections and decomposed arcs be well organized, which can facilitate the subsequent operations. Specifically,  
each item in the two sequence lists is a compound structure consisting of three domains: (\romannumeral 1) the  {processed related edge}; (\romannumeral 2) the  intersections (if exist) on this edge; and (\romannumeral 3) a tri-value switch. For ease of discussion, we denote by $S_1$ and $S_2$  the two sequence lists,  by $S_i[j]$ the $j$th item in $S_i$ ($i\in {1,2}$), and  by $S_i[j].a$, $S_i[j].b$ and $S_i[j].c$  the three domains of $S_i[j]$, respectively. 

\begin{figure}[h]  
\centering
      \includegraphics[scale=.6]{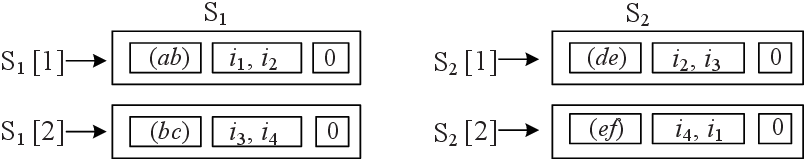} 
      \vspace{-1ex} 
 \caption{\small  Example of \textit{sequence lists}.} 
 \label{fig:sequence list}
\end{figure}

The {processed related edges} in each sequence list are stored in counter-clockwise direction with regard to the original circular-arc polygon. For example, regarding to circular-arc  polygons in Figure \ref{fig:many edges}, we  construct two sequence lists as shown in Figure \ref{fig:sequence list}.  
Note that when there are multiple intersections on an edge, we should keep these intersections in order. See $S_1[1].b$ of Figure  \ref{fig:sequence list} for an example, the point $i_1$ is ahead of the point $i_2$.    
Regarding to the third domain $S_i[j].c$, it  is  assigned to either 0,  1, or 2. The assignment rules are as follows. When the edge is not a decomposed arc, we assign ``0'' to $S_i[j].c$. In this example, for any  $1\leq j\leq |S_i|$ (where $|\cdot|$ denotes the cardinality of $S_i$), $S_i[j].c$ is set to 0, as there is no decomposed arc. Otherwise, we assign ``1'' or ``2'' to $S_i[j].c$.
The readers may be curious why we use two different values. The purpose is to  differentiate the decomposed arcs which are from  different non-x-monotone arcs. 
This can help us efficiently merge them in the future. (The specific steps on how to merge them will be discussed in Section \ref{sec:two new linked lists}.)  
Given a series of decomposed arcs,  we  assign ``1'' to each decomposed arc  that is from the \textit{odd} (1st, 3rd,  $\cdots$) non-x-monotone arc, and assign ``2'' to each  decomposed arc that is from the \textit{even} (2nd, 4th,  $\cdots$) non-x-monotone arc.  

See Figure \ref{fig:third domain a} for an example, there are   five   \textit{related edges} in $\mathscr P_1$. Furthermore,  Figure \ref{fig:third domain b} illustrates  eight  \textit{processed related edges} (after  we decompose them based on Lemma \ref{lemma:at least at most}), implying that  $|S_1|=8$. Based on the assignment rules,  the values of the third domains should be ``0, 1, 1, 2, 2, 1, 1, 0'', respectively.

So far, we have shown how to use two sequence lists to manage the processed related edges and intersections. Note that, in order to obtain the intersections, a standard technique   is the \textit{plane sweep method}  \cite{JonLouisBentley:Algorithms,MichaelIanShamos:Geometric}. In this paper, we do not directly use this algorithm. Instead, we modify it by adding \textit{two labels} to    avoid the ``false'' intersections being reported, and  to speed up the process of inserting the reported intersections into their corresponding edges. (Remark: here the \textit{false} intersections refer to the vertices of polygons.)  We next give a brief summary of the plane sweep algorithm, and then   show how the two labels work.

\textit{Plane sweep method.} 
Let $\mathscr Q$ be a priority queue,  $\mathscr R$ be a balanced tree{\small \footnote{It is not mandatory to use a  priority queue and a balanced tree, whereas they are usually being  recommended, for the sake of efficiency \cite{JonLouisBentley:Algorithms}. Moreover, both of them are  abstract concepts; the priority queue, for example, can be implemented with a heap or other methods.}}, and  $l_v$ be a vertical sweep line. The basic idea of the plane sweep method is  as follows. First, it sorts the endpoints of all segments according to their x-coordinates, and then puts them into  $\mathscr Q$. Next, it  sweeps the plane  (from left to right) using $l_v$. At each endpoint during this sweep, if an  endpoint is the left endpoint of a segment, it inserts this segment into  $\mathscr R$; in contrast, if it is the right endpoint of a segment, it deletes this segment from $\mathscr R$. Note that all the segments  intersecting with $l_v$  are stored (in order from bottom to top) in  $\mathscr R$. In particular, when  $l_v$ moves from one endpoint to another endpoint,   it always checks  whether or not newly adjacent segments  intersect with each other; If so, it computes the intersection. In this way, all intersections can be obtained finally{\small \footnote{Note that, in some cases the segments may be vertical line segments,  or the segments may be tangent, or many  segments possibly intersect at one point; for these degenerated cases,  please refer to the papers (e.g., \cite{JonLouisBentley:Algorithms,MichaelIanShamos:Geometric,GuntherGreiner:efficient,YongKuiLiu:AnAlgorithm,AriRappoport:AnEfficient})  for more details. Unless stated otherwise,  degenerated cases are processed using  existing techniques and/or a straightforward adaptation from existing techniques. We  no longer expatiate them  separately for saving space (as they are tedious, and are not the focus of the paper). }}.

\begin{figure}[t]
  \centering
  \subfigure[\scriptsize { } ]{\label{fig:third domain a}
     \includegraphics[scale=.35]{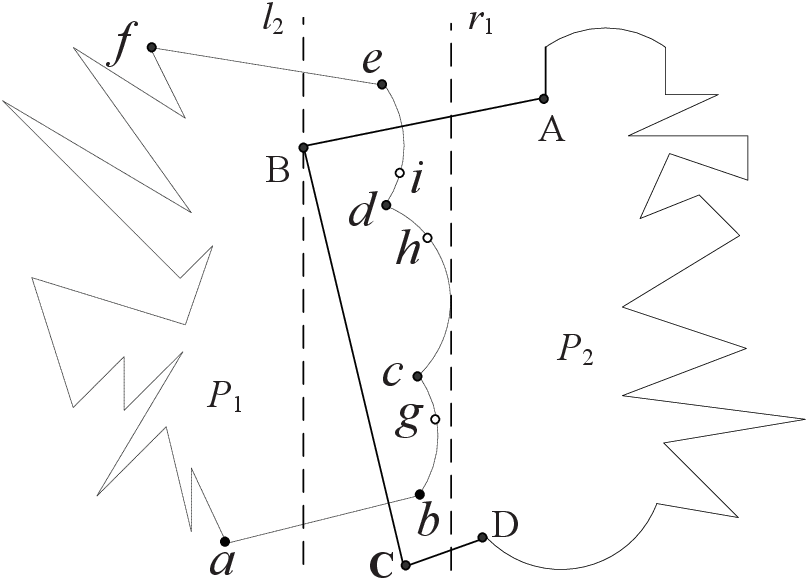}} 
     \hspace{-4ex}
  \subfigure[\scriptsize { }]{\label{fig:third domain b}
      \includegraphics[scale=.35]{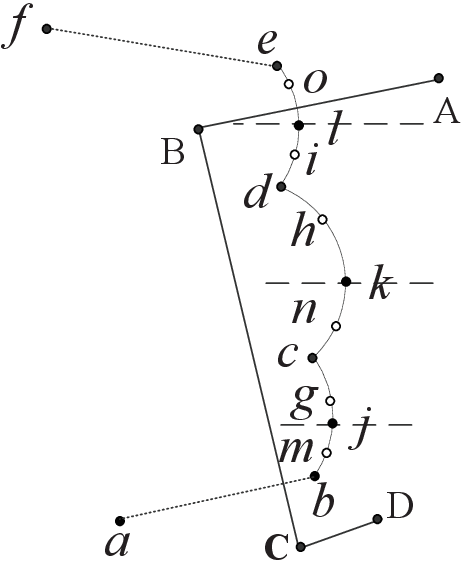}} 
      \vspace{-2ex} 
\caption{\small  Example of consecutive non-x-monotone circular arcs. (a) Edges $\overline{ab}$, $\widehat{b\widetilde{g} c}$, $\widehat{c\widetilde{h} d}$, $\widehat{d\widetilde{i} e}$, $\overline{ef}$  are  \textit{related edges} of $\mathscr P _1$. (b) Edges $\overline{ab}$, $\widehat{b\widetilde{m} j}$, $\widehat{j\widetilde{g} c}$, $\widehat{c\widetilde{n} k}$, $\widehat{k\widetilde{h} d}$, $\widehat{d\widetilde{i} l}$, $\widehat{l\widetilde{o} e}$, $\overline{ef}$  are  \textit{processed related edges of $\mathscr P _1$}; three dashed lines are the auxiliary lines. } 
 \label{fig:third domain}
\end{figure}

\vspace{-1ex}
\subsection{Avoiding False  Intersections and Speeding Up  Lookups}\label{subsec:avoidingFalse}

We can easily see that the  plane sweep method  directly inserts a segment into the balanced tree $\mathscr R$, if  the  point $p$ ($\in \mathscr Q$) is the left endpoint of the segment. Instead,  we assign two labels to the segment  before it is inserted into  $\mathscr R$.  Note that  the segment discussed  here  refers to the \textit{processed related edge}.
For clearness, we denote by $lb_1$ and $lb_2$ the two labels, respectively. 

$lb_1$ is the boolean type,  identifying that a segment is 
from which one of the two circular-arc polygons. Specifically,  if the segment is  from $\mathscr P_1$, we assign \textit{true} to $lb_1$; otherwise, we assign \textit{false} to $lb_1$.  
Recall that the plane sweep method always checks whether or not two segments 
intersect with each other, when they are adjacent. Our proposed method does not need to check them regardless of whether or not they intersect, if  the first labels of two adjacent edges have
the same value.
This  can avoid the unnecessary test and the ``false'' intersections.

$lb_2$ is an integer type denoting a serial number, which corresponds to the ``id''  of an item stored in the sequence list (note: the ``id'' information of each item is implied, as  we store the items using  the sequence list, i.e.,  array).   When  we detected an intersection, this label  can help us quickly find the item  in the sequence list,  and then insert the intersection into this item. 
See Figure \ref{fig:many edges b} for an example, $lb_1$ and $lb_2$  of edge $\overline{ab}$ are assigned to \textit{true} and   1, respectively. When we detected the intersection $i_1$, we thus can quickly know that we should insert the intersection into $S_1[1]$ (i.e., the first item of $S_1$). Otherwise,  we have to scan the sequence list in order to insert the intersection into an appropriate item, this way is inefficient especially when $|S_1|$ (or $|S_2|$) is large.

\smallskip

{\vspace{1ex}
\small  \hrule
\vspace{0.5ex}

\noindent \textbf{Algorithm 1} \textit{ConstructSequenceLists} 
\vspace{0.5ex}

\hrule
\vspace{0.5ex}
}

    {\footnotesize

\noindent \textbf{Input:} {Circular-arc polygons $\mathscr P _1$} and {$\mathscr P _2$}

 \noindent \textbf{Output:} {Sequence lists $S_1$ and $S_2$, related edge sets  $R_1$ and $R_2$}

\noindent $~~$1:$~~~~$Find the MBRs,   effective axis and   extended boundary lines

\noindent $~~$2:$~~~~$\textbf{for} each $i$ $\in \{1,2\}$ \textbf{do}

\noindent $~~$3:$~~~~$$~~~~$$R_i\leftarrow$ related edges from $\mathscr P_i$

\noindent $~~$4:$~~~~$Create two empty sequence lists $S_1$ and $S_2$

\noindent $~~$5:$~~~~$\textbf{for} each $i$ $\in \{1,2\}$ \textbf{do}

\noindent $~~$6:$~~~~$$~~~~$\textit{InitializeSequenceList} ($R_i$, $S_i$) // cf., Algorithm 2

\noindent $~~$7:$~~~~$Sort  the endpoints of  the segments (from $S_1,$ $S_2$),   and put them into 

\noindent $~~~~$$~~~~~~$the priority queue  $\mathscr Q$

\noindent $~~$8:$~~~~$Initialize the empty balanced tree $\mathscr R$

\noindent $~~$9:$~~~~$\textbf{for} each point $p$  $\in \mathscr Q$ \textbf{do}

\noindent 10:$~~~~$$~~$Let $s$ be the segment containing the point $p$, and $t$ be the  

\noindent $~~~~~~~~~~~~$segment  immediately  above or below $s$

\noindent 11:$~~~~$$~~$\textbf{if} ($p$ is the left endpoint of segment $s$) 

\noindent 12:$~~~~$$~~~~$$~~$Assign  two ``labels'' to $s$, and insert $s$ into $\mathscr R$

\noindent 13:$~~~~$$~~~~$$~~$\textbf{if} ( $s.lb_1$  $\neq t.lb_1$ )

\noindent 14:$~~~~$$~~~~$$~~~~$$~~$\textbf{if} ($s$  intersects with $t$) 

\noindent 15:$~~~~$$~~~~$$~~~~$$~~~~$$~~$Insert the intersection  into $\mathscr Q$, and also  insert into $S_1$ and $S_2$


\noindent 16:$~~~~$$~~$\textbf{else if} ($p$ is the right endpoint of  segment $s$) 

\noindent 17:$~~~~$$~~~~$$~~$\textbf{if} ($s.lb_1\neq t.lb_1$)

\noindent 18:$~~~~$$~~~~$$~~~~$$~~$\textbf{if} ($s$ intersects with $t$ and this intersection  $\notin \mathscr Q$)


\noindent 19:$~~~~$$~~~~$$~~~~$$~~~~$$~~$Insert this intersection  into $\mathscr Q$,   and   it into $S_1$ 

\noindent $~~~~~~~~~~~~~~~~~~~~~~~~~$also insert and $S_2$, respectively;  delete $s$ from $\mathscr R$   

\noindent 20:$~~~~$$~~$\textbf{else} // $p$ is an intersection  of two segments, say $s$ and $t$

\noindent 21:$~~~~$$~~~~$$~~$Swap the position of $s$ and $t$ //  assume $s$ is above $t$ 

\noindent 22:$~~~~$$~~~~$$~~$Let $t_1$ be the segment above $s$, and $t_2$ be the segment below $t$

\noindent 23:$~~~~$$~~~~$$~~$\textbf{if} ($s.lb_1\neq t_1.lb_1$ or $t.lb_1\neq t_2.lb_1$) 

\noindent 24:$~~~~$$~~~~$$~~~~$$~~$\textbf{if} ($s$ intersects with $t_1$, or $t$ intersects with $t_2$) 

\noindent 25:$~~~~$$~~~~$$~~~~$$~~~~$$~~$Insert the intersection point into $Q$, and

\noindent $~~~~~~~~~~~~~~~~~~~~~~~~~$ also insert it into $S_1$ and $S_2$, respectively 

\noindent 26:$~~~~$\textbf{return} $S_1$ and $S_2$, $R_1$ and $R_2$
\vspace{1ex}

\hrule
\vspace{1ex}
        }

 \vspace{-1ex}
\subsection{The Algorithm}\label{subsec:Step1Algorithm}

Let $R_1$ and $R_2$ be the related edges from $\mathscr P_1$ and $\mathscr P_2$, respectively. Given a segment $s$, we use $s.lb_1$ and $s.lb_2$ to denote the two labels of segment $s$. Algorithm 1  illustrates the pseudo-codes of  constructing the two sequence lists.

We first  choose the \textit{related edges} based on the extended boundary lines (Lines 1-3). Next, we construct two empty sequence lists and initialize them (Lines 4-6). After this,  we compute the intersections (Lines 7-25). In particular, when we compute the intersections, two  \textit{labels} are assigned to the segment before it is inserted into the balanced tree (Line 12), and we use the two  sequence lists to store the  intersections (Lines 15, 19 and 25). 
Note that, the pseudo-codes of initializing the two sequence lists are listed in Algorithm 2. This algorithm decomposes non-x-monotone arcs,  puts the \textit{processed related edges}   into the sequence lists in an orderly manner, and assigns appropriate values  to the tri-value switches.

\vspace{-1ex}

\begin{lemma}\label{theorem:modified plane sweep algorithm complexity}
Given two circular-arc polygons with $m$ and $n$ edges, respectively, and assume there are   $l$  \textit{related edges} between the two polygons, we have that  constructing the two sequence lists can be  finished in $O(m+n+l+(l+k)\log l)$ time, where $k$ is the number of intersections.
\end{lemma}

\begin{proof}
To obtain the \textit{related edges}, we first need to find the MBRs of two polygons, which takes linear time. We next determine the \textit{effective axis} by comparing the horizontal and vertical spans of the intersection set of two MBRs, which can be finished in constant time. Furthermore, the \textit{extended boundary lines} can be obtained in constant time once we obtain the effective axis. Based on two extended boundary lines, we finally obtain the \textit{related edges} by comparing the geometrical relationship between each edge and extended boundary lines, which also takes linear time.   Thus,  Lines 1-3 take $O(m+n)$ time.  

Creating two empty  sequence lists takes constant time. In addition, in order to initialize the two sequence lists, we need to decompose each non-x-monotone arc. Decomposing a single arc can be finished in constant time. In the worst case, all the \textit{related edges} are non-x-monotone arcs. Even so, there are no more than $3l$ items in the two sequence lists, according to Lemma \ref{lemma:at least at most}. Hence initializing two  sequence lists takes $O(l)$ time. Sorting all the endpoints of segments in the priority queue $\mathscr Q$  takes $O(l\log l)$ time, and initializing the balanced tree $\mathscr R$ takes constant time. Thus, Lines 4-8 take $O(l+l\log l)$ time.

As  there  are no more than $3l$ segments in $S_1$ and $S_2$, the number of endpoints thus is  no more than $6l$. So we can easily know that the number of executions of the  \textbf{for} loop (Line 9)  is no more than $6l+k$. 
Within the \textbf{for} loop, each operation (e.g., insert, delete, swap, find the above/below segment)  on $\mathscr R$ can be finished in $O(\log l)$ time, as the number of segments in $\mathscr R$ never exceeds $3l$.
Additionally, each of other operations (e.g., assign labels to the segment, determine if two segments intersects with each other)  can be finished in constant time. Thus, Lines 9-25 take $O((6l+k)\log l)$ time, i.e., $O((l+k)\log l)$ time. Putting it together, this completes the proof.  
\end{proof}

{\vspace{1ex}
\small  \hrule
\vspace{0.5ex}

\noindent \textbf{Algorithm 2} \textit{InitializeSequenceList} 
\vspace{0.5ex}

\hrule
\vspace{0.5ex}
}

    {\footnotesize

\noindent \textbf{Input:} { $R_i$, $S_i$ }

 \noindent \textbf{Output:} {$S_i$}

\noindent $~~$1:$~~$$temp\leftarrow 1$ // the $temp$ is used to set the tri-value switch 

\noindent $~~$2:$~~$\textbf{for} each related edge  $r\in R_i$ \textbf{do}

\noindent $~~$3:$~~~~$$~~$\textbf{if} ($r$ is a \textit{non-x-monotone circular arc} )

\noindent $~~$4:$~~~~$$~~~~$$~~$Decompose it and put the decomposed arcs into $S_i$, and 

$~~~~~~~~~~$set the value of each tri-value switch   to ``$temp$''

\noindent $~~$5:$~~~~$$~~~~$$~~$\textbf{if} ($temp$=1)

\noindent $~~$6:$~~~~$$~~~~$$~~~~$$~~$$temp\leftarrow$2;

\noindent $~~$7:$~~~~$$~~~~$$~~$\textbf{else} // $temp$=2

\noindent $~~$8:$~~~~$$~~~~$$~~~~$$~~$$temp\leftarrow$1;

\noindent $~~$9:$~~~~$$~~$\textbf{else} // $r$ is not a \textit{non-x-monotone circular arc}

\noindent 10:$~~~~$$~~~~$$~~$Put it into $S_i$, set  the value of tri-value switch  to ``0''

\hrule
\vspace{1ex}
        }

\smallskip

We have shown how to construct two sequence lists.   
It is easy to know that we cannot get the resultant polygon  based on  \textit{only} the information stored in the two sequence lists.   
Next, we are ready to merge the information in them and part of information in the original linked lists, and store the `merged information' using two \textit{new linked lists}. For ease of exposition, we call this step `building two new linked lists'.  Note that  the two new linked lists  will be significantly used in Section \ref{sec:traversing}, as we need to traverse them to get the resultant polygon.

\section{Building Two New Linked Lists}\label{sec:two new linked lists}

To  construct  two new linked lists, on the whole   we first initialize two (empty) new linked  lists, and then copy the information from the \textit{original linked lists} to the \textit{new linked lists} while we replace  those \textit{related edges}  using the information stored  in the two sequence lists.    
Note that there are two important yet easy-to-ignore issues  needed to be handled when we construct  new linked lists.
We next check these issues, and then present the algorithm of constructing new linked lists.


\vspace{-2ex}
\subsection{Eliminating The Ambiguity} 
Recall Section \ref{subsec:datastructure},  we always add an \textit{appendix point} between two vertices if an edge is a circular arc. When we replace \textit{related edges} with the information stored in sequence lists, we also have to ensure this property. It is easy to know that,  when the intersections appear on a circular arc, this arc will be decomposed by these intersections. We thus have to add the new appendix point for each sub-segment, in order to eliminate the ambiguity. 

\vspace{-1ex}

\begin{lemma}
Suppose there are  $k$  intersections on a circular arc; then, we  need to insert at least $k$  and at most  $k+1$ new \textit{appendix points}.
\label{lemma: insert new appendix point}
\end{lemma}

\begin{proof}
Since $k$ intersections can subdivide a complete circular arc into $k+1$ small circular arcs, and for each small circular arc one \textit{appendix point} is needed and  enough to eliminate the ambiguity. Clearly, $k+1$ \textit{appendix points} are needed for $k+1$ small circular arcs. Note that,  there is an \textit{appendix point} beforehand. Therefore,  when there is no any intersection (among all these intersections) that is  \textit{coinciding with} this existing  \textit{appendix point}, only $k$ new \textit{appendix points} are needed. Otherwise, we need $k+1$ new {\textit{appendix points}}.  
\end{proof}

Besides the above issue, another issue is to handle the decomposed arcs. We decomposed non-x-monotone arcs into x-monotone arcs ever, we thus need to merge them. The natural solution is to  compare each pair of adjacent edges of the resultant polygon, checking if they can be merged into a single arc. This way however is inefficient because (\romannumeral 1) most of edges of the resultant polygon may not need to be merged; and (\romannumeral 2) given two adjacent arcs, let $C_1$ and $C_2$ respectively denote their corresponding circles; checking if the two arcs can be merged into a single arc needs to compute the centres of $C_1$ and $C_2$, this will use  trigonometric functions (which could have been avoided). 
We next show how to efficiently merge them, with the help  of  the \textit{tri-value switch} (recall Section \ref{subsec:two sequence lists}).  


\vspace{-2ex}
\subsection{Efficiently Merging Decomposed Arcs}  
We merge the decomposed arcs when constructing new linked lists, rather than merge them after obtaining the resultant polygon. In particular,  we here utilize the information stored in the tri-value switch to improve the efficiency.
Specifically, given an item  $S_i[j]$, if  $S_i[j].c=1$ (or 2),  we continue to fetch its next item $S_i[j+1]$  from the  sequence list if  $S_i[j].c=S_i[j+1].c$. In this way,  a group of consecutive items   are fetched from the sequence list.
W.l.o.g, assume that we have fetched $\lambda$ consecutive items, $S_i[j]$,  $\cdots$, $S_i[j+\lambda-1]$.  Then, we do as follows.

\begin{itemize*}
\item If $S_i[j].b=S_i[j+1].b=\cdots=S_i[j+\lambda-1].b=\emptyset$, we discard the fetched items instead of merging them. This is because there is no intersection  on these decomposed arcs, the merged result should be the same as the edge in the original linked list. 
\item Otherwise, we  insert new appendix points,  merge decomposed arcs, and replace the edge in the original linked list.
\end{itemize*}

Let us revisit  Figure \ref{fig:third domain b}; recall that there are  eight  items in  $S_1$, and the values in their tri-value switches are ``0, 1, 1, 2, 2, 1, 1, 0'', respectively. Although $S_1[2].c=S_1[3].c=1$, we discard the two items  instead of merging them, as 
$S_1[2].b=S_1[3].b=\emptyset$. Similarly,  we also discard the items $S_1[4]$ and $S_1[5]$. Note that, for the 6th and 7th items, $S_1[6].c=S_1[7].c=1$ and $S_1[7].b\neq \emptyset$; thus, we   insert a new appendix point,  merge the  two decomposed arcs, and use the  merged result  to replace the edge in the original linked list. 

Note that, the consecutive items mentioned earlier are actually the decomposed arcs generated from a single non-x-monotone circular arc. According to  Lemma \ref{lemma:at least at most}, we can easily obtain the following  corollary.
\begin{corollary}\label{theorem:lambda consecutive items}
Let $\lambda$ be the number of consecutive items, we have that $\lambda\leq 3$ and $\lambda\geq 2$.
\end{corollary}

{\vspace{1ex}
\small  \hrule
\vspace{0.5ex}

\noindent \textbf{Algorithm 3} \textit{BuildNewLinkedLists} 
\vspace{0.5ex}

\hrule
\vspace{0.5ex}
}

    {\footnotesize

\noindent \textbf{Input:} { $\mathscr P_1$ and $\mathscr  P_2$, $S_1$ and $S_2$, $R_1$ and $R_2$ }

 \noindent \textbf{Output:} {$\mathscr P_1^*$ and $\mathscr P_2^*$}

 \noindent $~~$1:$~~~~$Set $\mathscr P_1^*=$$\mathscr P_2^*=\emptyset$, and $j\leftarrow 1$

 \noindent $~~$2:$~~~~$\textbf{for} each $i\in \{1,2\}$ \textbf{do}

 \noindent $~~$3:$~~~~$$~~~~$\textbf{for} each edge $e$  $\in \mathscr P_i$  \textbf{do}

 \noindent $~~$4:$~~~~$$~~~~$$~~~~$\textbf{if} ($e\notin R_i$ )

 \noindent $~~$5:$~~~~$$~~~~$$~~~~$$~~~~$Add $e$ to $\mathscr P_i^*$

 \noindent $~~$6:$~~~~$$~~~~$$~~~~$\textbf{else} // $e$ is  a \textit{related edge}


 \noindent $~~$7:$~~~~$$~~~~$$~~~~$$~~~~$\textbf{if} ($s_i[j].c= 0$) // not a decomposed arc

 \noindent $~~$8:$~~~~$$~~~~$$~~~~$$~~~~$$~~~~$\textbf{if} ($S_i [j].b=\emptyset$) // no intersection

 \noindent $~~$9:$~~~~$$~~~~$$~~~~$$~~~~$$~~~~$$~~~~$$j\leftarrow j+1$, and add $e$ to $\mathscr P_i^*$

 \noindent 10:$~~~~$$~~~~$$~~~~$$~~~~$$~~~~$\textbf{else} // $S_i [j].b\neq \emptyset$

 \noindent 11:$~~~~$$~~~~$$~~~~$$~~~~$$~~~~$$~~~~$\textbf{if} ($S_i [j]$ is a circular arc)

 \noindent 12:$~~~~$$~~~~$$~~~~$$~~~~$$~~~~$$~~~~$$~~~~$Insert new appendix points

 \noindent 13:$~~~~$$~~~~$$~~~~$$~~~~$$~~~~$$~~~~$Put the information  from $S_i [j]$ into $\mathscr P_i^*$, and set $j\leftarrow j+1$

 \noindent 14:$~~~~$$~~~~$$~~~~$$~~~~$\textbf{else}  // $s_i[j].c= 1$ (or 2)

 \noindent 15:$~~~~$$~~~~$$~~~~$$~~~~$$~~~~$Set  $tri=S_i[j].c$, and $\lambda \leftarrow 0$

 \noindent 16:$~~~~$$~~~~$$~~~~$$~~~~$$~~~~$\textbf{do} // copy the consecutive decomposed arcs

 \noindent 17:$~~~~$$~~~~$$~~~~$$~~~~$$~~~~$$~~~~$$\lambda\leftarrow \lambda+1$, $temp[\lambda]\leftarrow S_i[j]$, $j\leftarrow j+1$

 \noindent 18:$~~~~$$~~~~$$~~~~$$~~~~$$~~~~$\textbf{while} $S_i[j].c=tri$ 

 \noindent 19:$~~~~$$~~~~$$~~~~$$~~~~$$~~~~$\textbf{if} ($temp[1].b=\cdots=temp[\lambda].b=\emptyset$) 

 \noindent 20:$~~~~$$~~~~$$~~~~$$~~~~$$~~~~$$~~~~$Put  $e$ into $\mathscr P_i^*$

 \noindent 21:$~~~~$$~~~~$$~~~~$$~~~~$$~~~~$\textbf{else}

 \noindent 22:$~~~~$$~~~~$$~~~~$$~~~~$$~~~~$$~~~~$Insert new appendix points,  merge decomposed arcs, and 
 
$~~~~~~~~~~~~~~~~~~~~~~$ put the merged result into $\mathscr P_i^*$

 \noindent 23:$~~~~$\textbf{return} $\mathscr P_1^*$ and $\mathscr  P_2^*$

\hrule
\vspace{1ex}
        }

\subsection{The Algorithm}
Let $\mathscr P_1^*$ and $\mathscr P_2^*$ be the two new linked lists, respectively. Algorithm 3  depicts the pseudo-codes of constructing two new linked lists. For each edge $e$ in the \textit{original linked list},  we check  whether it is a \textit{related edge}. If so, we further check  whether $S_i[j]$ is a decomposed arc.  Lines 7-13 are used to process the case when  it is not a decomposed arc. In contrast, Lines 14-22 are used to handle the opposite case. In this case, we first fetch all the consecutive decomposed arcs (Lines 15-18), and then  check if there are intersections on these decomposed arcs. If it is not,  we   put the edge $e$ into $\mathscr P_i^*$ (Lines 19-20). Otherwise, we insert new appendix points, merge decomposed arcs, and put the merged result (instead of $e$) into $\mathscr P_i^*$ (Lines 21-22).

\begin{lemma}
Suppose we have obtained the two sequence lists $S_1$ and $S_2$; then,  constructing the two new  linked lists $\mathscr P_1^*$ and $\mathscr P_2^*$  takes $O(m+n+l+k)$ time.
\label{theorem: consturct modified polygon}
\end{lemma}

\begin{proof}
Inserting a single \textit{appendix point} takes constant time. In the worst case, all the intersections are located on arcs rather than on line segments. Even so,  there are  no more than $2k$ new \textit{appendix points} according to Lemma \ref{lemma: insert new appendix point}. Hence inserting \textit{appendix points}  takes  $O(k)$ time.   Merging  $\lambda$ consecutive decomposed arcs takes constant time, as $\lambda\leq 3$ (cf.,  Corollary    \ref{theorem:lambda consecutive items}). In the worst case, all the related edges are non-x-monotone arcs, hence merging all  consecutive decomposed arcs takes $O(l)$ time.  

Since the number of edges in $\mathscr P_1$ and $\mathscr P_2$ is $m+n$,   the number of executions  of the second \textbf{for} loop (Line 3) is $m+n$. Specifically,  the number of executions of Line 4 is $m+n-l$, and the one of Line 6 is $l$.   Even if all  related edges are  non-x-monotone arcs,  the number of executions of Line 17 is no more than $3l$. Furthermore, within the \textbf{for} loop, each  operation takes constant time (note: here we  no longer  consider the time for inserting new appendix points and merging decomposed arcs, as we have analysed them  in the previous paragraph). Therefore, Lines 4-5 and Lines 7-22 take $O(m+n-l)$ and $O(3l)$ time, respectively.  
Putting it all together, this completes the proof. 
\end{proof}

\section{Traversing}\label{sec:traversing}

 \vspace{-1ex} 
 
In the previous section, we have obtained $\mathscr P_1^*$ and $\mathscr P_2^*$. This section shows in detail how to get the resultant polygon by traversing them. 
In order to correctly traverse  $\mathscr P_1^*$ and $\mathscr P_2^*$, we need to assign  the \textit{entry}-\textit{exit} properties to  intersections.

\vspace{-2ex}

\subsection{Entry-Exit Property}\label{subsec:entryExit}
The entry-exit property is an important symbol that was ever used in many papers focusing on boolean operation of traditional polygons (see e.g., \cite{GuntherGreiner:efficient,YongKuiLiu:AnAlgorithm}). The followings show this technique can be  used to the case of our concern as well. 
Specifically, we  assign the intersections with the \textit{entry} or \textit{exit} property \textit{alternately}.    
Note that the \textit{entry-exit} property for the first intersection in $\mathscr P^*_1$ ($\mathscr P^*_2$) is determined as follows. W.l.o.g, assume the first intersection is $i$ ($i^\prime$) in $\mathscr P^*_1$ ($\mathscr P^*_2$), and the previous node of $i$ ($i^\prime$) is $i.prev$ ($i^\prime.prev$).   We check if  $i.pre$   ($i^\prime.pre$) is outside the input polygon $\mathscr P_2$ ($\mathscr P_1$). If so, we assign the \textit{entry} (\textit{exit}) property to $i$ ($i^\prime$).

Once the \textit{entry-exit} properties are assigned to intersections,  we then obtain the resultant polygon based on the traversing rules  below.

\vspace{-2ex}
\subsection{Traversing Rules}\label{subsec:traversing rule}
Let $i_s$ be an intersection (point) of $\mathscr P_1^*$ such that $i_s$ has the \textit{entry} property. Let $v_s$ be a vertex of $\mathscr P_1^*$ such that $v_s$ does not locate in $\mathscr P_2^*$. 
There are three typical boolean operations: intersection, union and difference. Note that in the rest of discussion, the default traversing direction is counter-clockwise, unless stated otherwise.

\textit{Intersection.} 
We start to  traverse  $\mathscr P_1^*$ using  $i_s$ as the starting point. Once we meet an intersection  with the {\textit{exit}} property, we shift to  $\mathscr P_2^*$, and  traverse it. Similarly, if we meet an intersection  with the {\textit{entry}} property in $\mathscr P_2^*$, we shift back to $\mathscr P_1^*$.  In this way, a circuit will be produced. After this, we check if there is another intersection  of $\mathscr P_1^*$ such that (\romannumeral 1) it has  the {\textit{entry}} property; and  (\romannumeral 2)   it  is not a vertex of the  produced circuit.
If no such an intersection,  we terminate the traversal, and  this circuit is the intersection between $\mathscr P_1$ and $\mathscr P_2$.
Otherwise,    we let this intersection  as a new starting point, and  traverse  the two new linked lists (using the same method discussed just now), until no such an intersection  exists. In the end, we get multiple circuits, which are the intersection between  $\mathscr P_1$ and $\mathscr P_2$.

\textit{Union.} 
For this case, we, however, traverse  $\mathscr P_1^*$ using  $v_s$ as the starting point.
Once we meet an intersection  with the \textit{entry} property,  we shift to $\mathscr P_2^*$, and  traverse  it. Similarly, if we meet an intersection  with the \textit{exit} property in $\mathscr P_2^*$, we shift back to $\mathscr P_1^*$. In this way, a circuit will be produced, which is  the union between  $\mathscr P_1$ and $\mathscr P_2$.

\textit{Difference.} 
The first several steps are the same as the ones in  the \textit{union} operation, \textit{but} we traverse $\mathscr P_2^*$   in   \underline{clockwise}  direction. Similarly, if we meet an intersection  with the \textit{exit} property in $\mathscr P_2^*$, we shift back to $\mathscr P_1^*$. In this way, a circuit will be produced. Furthermore, we check if there is another vertex of $\mathscr P_1^*$ such that (\romannumeral 1) it is not a vertex of any  produced circuit; and (\romannumeral 2) it does not locate in $\mathscr P_2^*$. 
If no such a vertex,  we terminate the traversal, and this circuit is the difference between $\mathscr P_1$ and $\mathscr P_2$. 
Otherwise,  we let the vertex as a new starting point, and  traverse  the two new linked lists (using the same method discussed just now), until no such a vertex exists. In the end, we get multiple circuits, which are the difference between $\mathscr P_1$ and $\mathscr P_2$.  

\vspace{-1.5ex}

\subsection{The Algorithm}
\label{subsec:traversingAlgorithm}
In some cases the  result   consists of multiple circuits,  we denote by  $l_j$  the  linked list used to store the $j$th circuit. Let $n_d$ be a node of  $\mathscr P^*_i$ where $i\in \{1,2\}$, we  denote by $c_n$ the current node when we traverse $\mathscr P^*_i$, and denote by $c_n.next$ the next node of $c_n$. Furthermore, we denote by $i_s^\prime$ the  intersection  of $\mathscr P_1^*$ such that $i_s^\prime$ has the entry property and it is not a vertex of any produced circuit, and we use $\exists (i_s^\prime)=true$ to denote that  there exists such a point. 
Algorithm 4  illustrates the pseudo-codes of obtaining  the intersection result (remark: the ones of other two operations   can be  constructed similarly by traversing rules,   omitted).

\begin{lemma}
Given the two new  linked lists $\mathscr P_1^*$ and $\mathscr P_2^*$, to obtain the resultant polygon  takes $O(k+m+n+l)$ time.
\label{thorem: obtain resultant polygon}
\end{lemma}

\begin{proof}
Assigning the \textit{entry-exit} property to each intersection  takes constant time, and there are $k$ intersections  on each new  linked list. Thus, assigning  \textit{entry-exit} properties to  intersections   takes  $O(k)$ time.

$\mathscr P_1^*$ and $\mathscr P_2^*$ are used to generate the resultant polygon, they store the vertices, \textit{appendix} points, and intersections. The number of vertices is $2(m+n)$. In the worst case, all edges of two input polygons are circular arc segments, implying that the number of \textit{appendix} points in the input polygons is $m+n$; in this case, all the $k$ intersections are located on arcs, we  need to insert new \textit{appendix} points, and the number of new \textit{appendix} points  is no more than $k+1$, according to Lemma \ref{lemma: insert new appendix point}. So the number of all \textit{appendix} points in $\mathscr P_1^*$ and $\mathscr P_2^*$ is no more than $m+n+k+1$. Therefore, the total number of nodes in $\mathscr P_1^*$ and $\mathscr P_2^*$ is no more than $3(m+n)+2k+1$. 
Further, each operation on a node (e.g., determine the type of a node, insert a node into the resultant polygon)  takes constant time. Therefore, the traversal  takes $O(3(m+n)+2k+1)$ time. 
Putting it all together,  we have that   obtaining the resultant polygon takes $O(m+n+k+l)$ time when $\mathscr P_1^*$ and $\mathscr P_2^*$ are given beforehand{\small \footnote{It is simple to determine  $i_s^\prime$ (cf., Lines 8 and 22). We just need to collect all the intersections with entry properties in a data structure when we assign  entry-exit properties to intersections, and then remove the intersections from this data structure once they have already been visited in the process of traversing. After a circuit is formed, we  check if this data structure is empty. If otherwise, any intersection stored in this data structure can be taken as $i_s^\prime$.}}.
\end{proof}

{\vspace{1ex}
\small  \hrule
\vspace{0.5ex}

\noindent \textbf{Algorithm 4} \textit{TraverseLinkedLists} 
\vspace{0.5ex}

\hrule
\vspace{0.5ex}
}

    {\footnotesize

\noindent \textbf{Input:} { $\mathscr P_1^*$ and $\mathscr  P_2^*$ }

 \noindent \textbf{Output:} {$\mathscr P_3$}

 \noindent $~~$1:$~~~~$Set $j=1$

 \noindent $~~$2:$~~~~$\textbf{for} each  $i\in \{1,2\}$ \textbf{do}
 
\noindent $~~$3:$~~~~$$~~~~$Assign entry-exit property to $\mathscr P_i^*$

\noindent $~~$4:$~~~~$\textbf{do}

\noindent $~~$5:$~~~~$$~~~~$\textbf{if} ($j$=1)

\noindent $~~$6:$~~~~$$~~~~$$~~~~$Choose a starting point $i_s$ from $\mathscr P^*_1$

\noindent $~~$7:$~~~~$$~~~~$\textbf{else} 

\noindent $~~$8:$~~~~$$~~~~$$~~~~$Let $i_s\leftarrow i_s^\prime$ // i.e., let  $i_s^\prime$ be a new starting point

\noindent $~~$9:$~~~~$$~~~~$Set $c_n\leftarrow i_s$, and $l_j=\emptyset$ 

\noindent 10:$~~~~$$~~~~$\textbf{do}

\noindent 11:$~~~~$$~~~~$$~~~~$Put $c_n$ into $l_j$

\noindent 12:$~~~~$$~~~~$$~~~~$\textbf{if} ($c_n.next$ is not an intersection point)

\noindent 13:$~~~~$$~~~~$$~~~~$$~~~~$Let $c_n\leftarrow c_n.next$, and put $c_n$ into $l_j$

\noindent 14:$~~~~$$~~~~$$~~~~$\textbf{else}

\noindent 15:$~~~~$$~~~~$$~~~~$$~~~~$Shift to $\mathscr P_2^*$, choose  the node $n_d$ such 
that  $n_d=c_n.next$, let

$~~~~~~~~~~~~~~~$ $c_n\leftarrow n_d$, and put  $c_n$ into $l_j$

\noindent 16:$~~~~$$~~~~$$~~~~$$~~~~$\textbf{if} ($c_n.next$ is not an intersection point)

\noindent 17:$~~~~$$~~~~$$~~~~$$~~~~$$~~~~$$c_n\leftarrow c_n.next$, and put $c_n$ into $l_j$

\noindent 18:$~~~~$$~~~~$$~~~~$$~~~~$\textbf{else} 

\noindent 19:$~~~~$$~~~~$$~~~~$$~~~~$$~~~~$Shift to $\mathscr P_1^*$, choose  the node $n_d$ such 
that  $n_d=c_n.next$, and

$~~~~~~~~~~~~~~~~~~$ let $c_n\leftarrow n_d$

\noindent 20:$~~~~$$~~~~$\textbf{while} ($c_n\neq$ $i_s$)

\noindent 21:$~~~$$~~~~$Let $\mathscr P_3\leftarrow \mathscr P_3\cup l_j$, and set $j\leftarrow j+1$ 

\noindent 22:$~~~~$\textbf{while} ($\exists (i_s^\prime)=true$)

\noindent 23:$~~~~$\textbf{return} $\mathscr P_3$

\hrule
\vspace{1ex}
        }

\smallskip

\smallskip

Up to now, we have addressed all the main steps of our algorithm --- R{\small E2L}. We next analyse its  complexity. 

\vspace{-1ex}
\section{Time/Space Complexity}\label{sec:complexity}
We  analyse the  complexity of our algorithm \textit{using}  the intersection operation as a sample (note:   the complexity of other two operations is  the same as the one of this operation, and can be derived similarly, omitted). 

\begin{theorem}\label{theorem:time space complexity}
Given two circular-arc polygons with $m$ and $n$ edges, respectively, and assume there are   $l$  \textit{related edges} between the two circular-arc polygons. Then, to achieve  boolean operation on them takes $O(m+n+(l+k)\log l)$ time, using $O(m+n+k)$ space, where $k$ is the number of intersections. 
\end{theorem}

\begin{proof}
Our algorithm consists three main steps, and they take time $O(m+n+l+(l+k)\log l)$, $O(m+n+k+l)$,  and $O(m+n+k+l)$, respectively (see Lemmas \ref{theorem:modified plane sweep algorithm complexity},  \ref{theorem: consturct modified polygon} and  \ref{thorem: obtain resultant polygon}).   
Putting these results  together, hence the  time complexity is  $O(m+n+(l+k)\log l)$.

The space used in our algorithm mainly consists of  two groups of  related edges $R_1$ and $R_2$,  two sequence lists $S_1$ and $S_2$,  the priority queue $\mathscr Q$,   two new linked lists $\mathscr P_1^*$ and $\mathscr P_2^*$, and the balanced tree $\mathscr R$. (Remark: the input polygons $\mathscr P_1$, $\mathscr P_2$, and the output polygon  $\mathscr P_3$ are the input and output data; by the convention{\small \footnote{As an example,  the \textit{bubble sort algorithm} takes $O(1)$ space for sorting arbitrary $n$ natural numbers.}}, we  do not need to  consider them when we analyse the space complexity.)

Specifically, $R_1$ and $R_2$ have the size of  $O(l)$, as they are used to store the \textit{related edges}.  $S_1$ and $S_2$ are used to store the \textit{processed related edges}. In the worst case, the number of \textit{processed related edges} is no more than $3l$, according to Corollary \ref{lemma:processed related edges}. So $S_1$ and $S_2$ have the size of $O(3l)$.  Recall Algorithm 1, the priority queue $\mathscr Q$ is used to store the  endpoints of \textit{processed related edges} and the intersections, and so $\mathscr Q$ has the size of $O(6l+k)$. Regarding to the balanced tree $\mathscr R$, it has the size of $O(l)$ at most, as it is used to store the segments \textit{currently} intersecting the sweep line.   Furthermore,  $\mathscr P_1^*$ and $\mathscr P_2^*$ have the size of $O( 3(m+n)+2k+1)$, see the proof of Lemma  \ref{thorem: obtain resultant polygon}.  Putting it all together, we have that the  space complexity of our algorithm is $O(m+n+k+l)$. In addition,  the upper bound of the parameter $l$ is $m+n$. This completes the proof. 
\end{proof}

We can see that our algorithm roughly consumes linear space when $k$ is small.  The running time also approximates to linear complexity when $l$ and $k$ are small. It is noteworthy that a straightforward adaptation from the plane sweep algorithm (see e.g., the `Standard' method described in Section \ref{subsec:meths}) requires  $O( (m+n+k)\log (m+n)  )$ time. In other words, even if $k$ is small (e.g., $k<<m+n$), it is also \textit{linearithmic time}. We remark that, for boolean operations on circular-arc polygons,  the $O( (m+n+k)\log (m+n)  )$ result is actually the state-of-the-art competitor in terms of computational complexity.

Although our algorithm has some advantages to some extent, we have to point out that in the worst case (note: this case is possible although it is not the usual case), i.e., $l=m+n$, the running time deteriorates to $O( (m+n+k)\log (m+n)  )$, which is equal to the one of the standard method. In this case, the advantages of the proposed algorithm disappear, viewed from the theoretical perspective.   To this step, an interesting issue is: when $l=m+n$, whether or not its practical efficiency is also the same as the one of the standard method? We will experimentally evaluate our algorithm as well as the competitors in Section \ref{sec:experiment}, after we introduce some immediate extensions.

\section{Extensions}\label{sec:extension}

While this paper focuses on boolean operation of circular-arc polygons, our techniques can be easily extended to compute boolean operation of  traditional polygons. Assume  there are two traditional polygons, for example, we can also use two  lists to represent them. In this case, the  {\small \textsf{Tag}} domain is unneeded as the traditional polygons do not need the appendix points. We can also choose  \textit{related edges} based on the extended boundary lines, and store them using \textit{two sequence lists}. Note that  the third domain of the sequence list is unneeded, as here all related edges  are straight line segments.  When computing the intersections, we can also use  \textit{two labels} to speed up the process of inserting the intersections into their corresponding edges, and to avoid \textit{false} intersections. Next, we also construct \textit{two new linked lists}, by using the information stored in two sequence lists to replace the related edges in the original linked lists. Particularly, we here do not need to insert new appendix points and merge the decomposed arcs, as the traditional polygons have no such information. We finally get the resultant polygon by \textit{traversing} the two new linked lists, it is the same as that in Section \ref{sec:traversing}.

Furthermore, the discussions presented in previous sections assumed the circular-arc polygons to be operated have no holes. If we want to handle the opposite case, this can  be easily achieved by a straightforward adaptation of our proposed method.  Assume  we want to compute the  intersection of two circular-arc polygons with holes, for example,  we can  use multiple  lists to represent the circular-arc polygon with holes. One is to represent the outer boundaries of the circular-arc polygon, others are respectively to represent the boundaries of each hole. We can 
first compute the intersection of  the outer boundaries of two polygons, and then use this intersection result to subtract each hole of the two polygons.  All of these steps are quite straightforward, when our proposed method is given beforehand.

Finally  it is also immediate  to compute boolean operation on  circular-arc polygons with \textit{self-intersection}. Assume we want to compute the intersection, for example, we just need to make a minor modification on the traversing rule presented in Section \ref{subsec:traversing rule}. We also start to traverse $\mathscr P_1^*$ using an intersection with the \textit{entry} property as the starting point, and shift to $\mathscr P_2^*$ when  the number of intersections we meet in $\mathscr P_1^*$  is even. The main difference is that  the \textit{traversing direction} here is not always constant. 
To determine whether or not the traversing direction is needed to change when we shift from $\mathscr P_1^*$ ($\mathscr P_2^*$) to $\mathscr P_2^*$ ($\mathscr P_1^*$),  a key step is to check if the \textit{entry-exit} property of this intersection in $\mathscr P_1^*$ is different from the one in $\mathscr P_2^*$.  If so, we need to change the traversing direction. Otherwise, we needn't.

\section{Performance Evaluation}\label{sec:experiment}
This section  evaluates  our algorithm  experimentally. Specifically, Section   \ref{subsec:meths} describes the baselines. The experimental settings and datasets are introduced in Section \ref{subsec:settingAndDataset}, and the   results are reported in Section \ref{subsec:comparative result}. We also investigate  our proposed algorithm based on a specific application (Section \ref{subsec:caseStudy}).

\subsection{Methodologies }\label{subsec:meths}

We compare our method (i.e., R{\footnotesize E2L}) with several baselines,  which are either the algorithms used to handle more general case of polygons, or the simpler versions of our proposed method.  We  shortly introduce  them as follows.

\begin{table}[t]
\begin{center}
\vspace{0ex}
\begin{tabular}{p{.05\textwidth}| p{.4\textwidth} } \hline 
{{\scriptsize  Polygon}}&{{\scriptsize $~~~~~~~~~~~~~~~~~~~~~~~~~~$ Coordinates}} \\\hline
{\scriptsize $\mathscr P_1$}  & {\scriptsize (10,10), (40,10), (40,30), (\underline{32.5,40}), (20,40),(15,30), (\underline{25,22.5}),(15,15)}	 \\

{\scriptsize $\mathscr P_2$}  &{\scriptsize  (20,20), (\underline{32.5,25}) (45,20), , (55,30), (\underline{47,38.625}), (50,50), (30,45)}	\\ 

\hline
\end{tabular}
\end{center}
\vspace{-2ex}
\caption{\small  The coordinates of vertices are listed  counter-clockwise, and the left-bottom vertex is listed at first. The values tagged with  ``$\_$'' denote the coordinates of \textit{appendix points}.  }\label{tab:vertex coordinators}
\end{table}


\textit{CGAL.} We directly use the implementation of C{\footnotesize GAL}.$~$The essence of this method   is to directly$~$invoke the algorithm of boolean operation on  \textit{general polygons},$~$defined$~$as   \texttt{GeneralPol} \texttt{ygon\_2} in C{\footnotesize GAL}.  Specifically, the  ``CGAL::Cartesian$<$Number$\_$ type$>$''$~$is$~$used$~$as$~$the$~$kernel, in$~$which$~$``Number$\_$type''$~$denotes the$~$exact$~$rational$~$number-type$~$by default (see the header file ``arr$\_$rational$\_$nt.h'' in  C{\footnotesize GAL} for reference).   Based on this kernel, we construct ``CGAL::Gps$\_$circle $\_$segment$\_$traits$\_$2'' trait class, and the following objects ``Curve$\_$2,  X$\_$monotone$\_$curve$\_$2, Ge neral$\_$ polygon$\_$2, Point$\_$2''  are used, which  inherit the  trait class above.  




 \begin{table*}[t]
 \begin{center}
 \vspace{0ex}
 \begin{tabular}{p{.05\textwidth}| p{.93\textwidth} } \hline 
 {{\scriptsize  Polygon}}&{{\scriptsize $~~~~~~~~~~~~~~~~~~~~~~~~~~$ Coordinates}} \\\hline
 {\scriptsize $\mathscr P_1$}  & {\scriptsize  (15,5),(88.5,16.5),(\underline{80.95,25.9}),(69,27.5),(60,27.5),(60,35),(70,35),(\underline{72.35,43.05}),(68.5,50.5),(4.5,40.5),(\underline{13.65,33.65}),(25,35),(20,30), (28.5,27.5), (\underline{22.4,26.75}), (18.5,22)}	 \\
 
 {\scriptsize $\mathscr P_2$}  &{\scriptsize   
 (33.5,21),(\underline{39.5,17.15}),
 (46.5,18.5),(46.5,25.5),(56,12.5),(70,22),(64,22),(\underline{67.9375,31.9643}),
 (66,42.5),(36,37),(21.5,42.5),(\underline{17.4643,29.32145}),
 (21,16) }	\\ 
 
 \hline
 \end{tabular}
 \end{center}
 \vspace{-2ex}
 \caption{\small  The coordinates of vertices are listed  counter-clockwise, and the left-bottom vertex is listed at first.  }\label{tab:vertex coordinators of four test}
 \end{table*}

\textit{Berberich.} We directly use the algorithm in \cite{EricBerberich:aComputational}, which is initially developed for computing  boolean operations of conic polygons.  This method employs the D{\footnotesize CEL}  structure to represent the polygons. It first decomposes   non-x-monotone conic curves  and then  computes the arrangement of segments using the plane sweep method; next, it uses the results of arrangement to compute the \textit{overlap} of two polygons in order to achieve boolean operation. This algorithm is similar  
to that of C{\footnotesize GAL}, recall Section \ref{sec:related work}. (Remark: more information about the D{\footnotesize CEL} structure and computing the overlap of two polygons can be found in  \cite{MarkdeBerg:computational}.)   

\textit{Na\"{i}ve.}  It is one of simpler versions of our proposed method. This  method  employs  our proposed data structure, it however computes the intersections by comparing each pair of edges. Specifically, for each edge $e$ of $\mathscr P_1$, it checks whether $e$ intersects with the edges of $\mathscr P_2$. If so, it computes the intersections and inserts them into corresponding edges. It does in this way, until all edges of $\mathscr P_1$ are processed. The rest of steps are to assign  the \textit{entry-exit} properties  and to traverse, which are the same as the ones of our proposed method. Note that Greiner-Hormann's algorithm \cite{GuntherGreiner:efficient}, that was initially developed for boolean operations of traditional polygons, also computes the intersections by comparing each pair of edges, and the rest of steps also include traversing. In this regard, the Na\"{i}ve method  can be also looked as a generalization   of Greiner-Hormann's method.

\textit{Standard.} It is also a simpler version of our proposed method,  but it  is different from the previous one. Specifically, it employs the standard plane sweep algorithm to compute the intersections instead of checking each pair of edges. Note that it does not adopt the proposed optimization strategies (e.g., `using related edges', `avoiding false intersections',  `speeding up the lookups', `speeding up the merging of decomposed arcs'), others  are the similar as the R{\small E2L} method. (Remark: the idea of the Standard method is roughly same to  that of \textit{C{\footnotesize GAL}} and \textit{Berberich}, but it simply gets rid of the generality and employs a targeted data structure.) 

\begin{figure}[t]
  \centering
  \subfigure[\scriptsize ]{\label{fig:usecase1}
      \includegraphics[scale=.45]{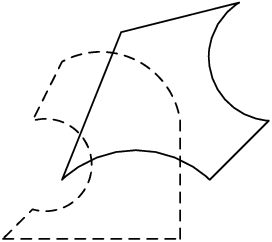}} 
   \subfigure[\scriptsize ]{\label{fig:figtest3}
       \includegraphics[scale=.4]{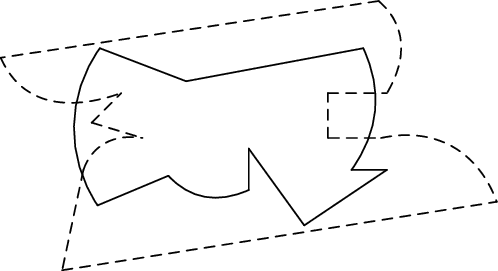}}   
     \subfigure[\scriptsize ]{\label{fig:usecase2}
       \includegraphics[scale=.4]{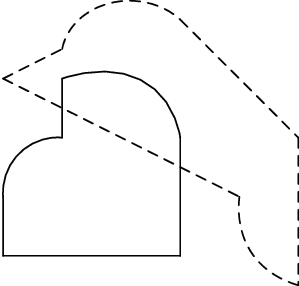}}
            \vspace{-2ex} 
\caption{\small  The use cases for our experiments.  The  polygon with dashed lines is $\mathscr P_1$, and  another one is $\mathscr P_2$. (a) For the first set of experiments. (b) For the fourth set of experiments. (c) A sampled example for which all algorithms work well, even if we use floating point data type.} 
 \label{fig:exp:8new}
\end{figure}



\subsection{Experimental Settings $\&$ Datasets}\label{subsec:settingAndDataset}
\subsubsection{Settings}
All the algorithms  are implemented in C++ language, the versions of L{\scriptsize EDA}, C{\scriptsize GAL} and B{\scriptsize OOST} library are  6.2,  4.3, and 1.46.1, respectively.   
The proposed algorithm and its simpler versions   do not employ  other libraries except the \textit{standard template library} (STL). The priority queue and the balanced tree mentioned in previous sections are implemented using  a heap and a red-black tree, respectively. The experiments are   conducted on a computer with 2.16GHz dual core CPU and 1.86GB of memory. By convention, we use  the execution time to measure the efficiency.  In our experiments,  we run 100 times by default for each algorithm and then compute the  average running time. 

\subsubsection{Datasets}\label{subsec:datasets}
\textit{Experiment 1.} We  manually produce two  circular-arc polygons, each of them is less than 10 edges,  for simplicity and for ease of reproducing the findings.  The vertex  coordinates  of the two polygons (cf., Figure \ref{fig:usecase1}) are listed in Table \ref{tab:vertex coordinators}. 
  
\textit{Experiment 2.} To  study the overall performance of these algorithms, we adopt thousands of circular-arc polygons.  Specifically, given an integer $n$,  a pair of circular-arc polygons with $n$ edges are generated at random{\small \footnote{Generally speaking, we first randomly generate two  rectangles such that they are overlapping each other. Then,  we  randomly generate $n$ points   in \textit{each} rectangle  one by one such that they satisfy two constraints: (\romannumeral 1) the segment joining  the $j$th point and the $(j-1)$th one cannot intersect with any other  segment except the segment joining the $(j-1)$th point and the $(j-2)$th one, where $j\geq 4$; and (\romannumeral 2) the segment joining the $1$st point and the $n$th one cannot intersect with any other segments except its two adjacent segments. These $n$ points will be used as the vertices of the circular-arc polygon. Finally, we import circular-arc segments by inserting a set of \textit{appendix points}.}}, and then each algorithm is executed alternately, using the pair of polygons as the input.  This is done one thousand times. In each trial,  the running time of each algorithm is  recorded, and accumulated to previous trials.  We  compute the average value for estimating the average running time of each algorithm. Furthermore, we vary the value of $n$, and obtain the running time of each algorithm using the same method mentioned above.


\textit{Experiment 3.} We use the parameter ``\texttt{double}'' to replace   the parameter    ``Number$\_$type'' in the kernel ``CGAL::Cartesian $<$Number$\_$type$>$'' (the 1st approach), and also  the parameter ``Rational'' in the kernel ``CGAL::Cartesian $<$Rational$>$'' (the 2nd approach), in order to investigate the performance  when all these algorithms use the floating point number-type. We randomly generate pairwise  circular-arc polygons as the test data. 

\textit{Experiment 4.} To answer the interesting issue mentioned in Section \ref{sec:complexity},  we use the  polygons satisfying the condition $l=m+n$ as the input, see Figure \ref{fig:figtest3} (it is a sample). The vertex coordinates of these two polygons are listed in Table \ref{tab:vertex coordinators of four test}. 

\begin{table}[t]
\begin{center}
\vspace{0ex}
\begin{tabular}{p{.08\textwidth} |p{.05\textwidth}  p{.05\textwidth} p{.05\textwidth} p{.05\textwidth} p{.05\textwidth} p{.05\textwidth} } \hline 
{{\scriptsize  Method}}&{{\scriptsize CGAL }}&{{\scriptsize Berberich}} & {{\scriptsize Naive }}& {{\scriptsize Standard }}&  {{\scriptsize RE2L }} \\\hline
{\scriptsize Time (sec.)} 		&{\scriptsize 0.0273 } & {\scriptsize 0.0287}  &{\scriptsize 0.00239} &{\scriptsize 0.00175} &{\scriptsize 0.00112} \\
{\scriptsize Impr. fac.}  		&{\scriptsize 24.375 } & {\scriptsize 25.625} &{\scriptsize 2.13} &{\scriptsize 1.5625} &{\scriptsize ---} \\
\hline
\end{tabular}
\end{center}
\vspace{-2ex}
\caption{\small  The average running time in the first set of experiments.   }\label{tab:result of first set of experiment}
\end{table}

\textit{Experiment 5.} Inspired by the curiosity,  we also investigate the scalability of our proposed algorithm  using polygons with a larger number of edges.  Note that in this set of experiments almost all the polygons  generated are self-intersection polygons when the number of edges  is equal to or larger than $100$. This is because  it is pretty difficult to generate polygons without self-intersections when the number of edges is large,  using the method mentioned in Footnote 7. Specifically, in this case  we do not employ the constraint (\romannumeral 2), see Footnote 7.

\subsection{Experimental Results}\label{subsec:comparative result}

\subsubsection{Results of The First Experiment}
\label{subsubsec:lessEdge}

Table \ref{tab:result of first set of experiment} lists the  results of Experiment 1.  Specifically, the methods are listed in Row one,  the average running time of each method is  listed in Row two, and the improvement factors{\small \footnote{Here the improvement factor refers to the ratio of time. Assume that the exectution time of the `Standard' method is 0.2 seconds, the one of the proposed method is 0.05 seconds, for example, the improvement factor is $\frac{0.2}{0.05}=4$.}} of our algorithm over the baseline methods are shown in Row 3.  
From this table, we can   see  that the proposed method outperforms its  simpler versions, demonstrating the effectiveness of our proposed strategies. Interestingly, the simpler versions of the proposed method yet  outperforms  the former two methods, let alone the proposed method. To some extent this verifies our previous claim ---  directly executing existing algorithms used to compute boolean operation of conic and/or general polygons is usually not efficient enough. 

Compared to the former two methods, although the latter three ones  adopt a different data structure, but we note that  the `Standard' method   also directly employs the plane sweep method, similar to the former two ones. Viewed from the theoretical perspective, the former two methods  should  have the  performance similar to the ``Standard'' method. To further verify this phenomenon and explain it, we hence conduct another set of experiments, evaluating the overall performance based on larger datasets.

\subsubsection{Larger Datasets}

\begin{table}[t]
\begin{center}
\vspace{0ex}
\begin{tabular}{p{.03\textwidth} p{.05\textwidth}  p{.05\textwidth} p{.05\textwidth} p{.05\textwidth} p{.07\textwidth} p{.07\textwidth} } \hline 
{{\scriptsize  $n$}}&{{\scriptsize CGAL }}&{{\scriptsize Berberich}} & {{\scriptsize Naive }}& {{\scriptsize Standard }}&  {{\scriptsize RE2L }} \\\hline
{\scriptsize 5} 		&{\scriptsize 0.0281 } & {\scriptsize 0.0293}  &{\scriptsize 0.00234} &{\scriptsize 0.0018} &{\scriptsize 0.00107} \\

{\scriptsize 10} 	&{\scriptsize 0.0609} &{\scriptsize 0.0772}	 &{\scriptsize 0.0062}&{\scriptsize 0.0041} &{\scriptsize 0.00256}\\ 

{\scriptsize 20} 	&{\scriptsize 0.1282}  &{\scriptsize 0.1297}	 &{\scriptsize 0.0125} &{\scriptsize 0.0072} &{\scriptsize 0.00369} \\

{\scriptsize 30} &{\scriptsize 0.1656}  &{\scriptsize 0.1741}		 &{\scriptsize 0.0328} &{\scriptsize 0.0176} &{\scriptsize 0.00614} \\ 

{\scriptsize 40}  	&{\scriptsize 0.5172}&{\scriptsize 0.5672}  &{\scriptsize 0.0391} &{\scriptsize 0.0182} &{\scriptsize 0.00683} \\ 

{\scriptsize 50} 	&{\scriptsize 0.61}  &{\scriptsize 0.681} &{\scriptsize 0.0594} &{\scriptsize 0.0213} &{\scriptsize 0.00851} \\ 
\hline
\end{tabular}
\end{center}
\vspace{-2ex}
\caption{\small  The average running time of each algorithm, where $n$ denotes the number of edges of each polygon.   }\label{tab:experiment_result2}
\end{table}

\begin{table}[t]
\begin{center}
\vspace{0ex}
\begin{tabular}{p{.08\textwidth} |p{.05\textwidth}  p{.05\textwidth} p{.05\textwidth} p{.05\textwidth} p{.05\textwidth} p{.05\textwidth} } \hline 
{{\scriptsize  Method}}&{{\scriptsize CGAL }}&{{\scriptsize Berberich}} & {{\scriptsize Naive }}& {{\scriptsize Standard }}&  {{\scriptsize RE2L }} \\\hline
{\scriptsize Time (sec.)}  		&{\scriptsize 0.0112 } & {\scriptsize 0.0127} &{\scriptsize 0.00192} &{\scriptsize 0.00151} &{\scriptsize 0.000948} \\
{\scriptsize Imp. fac.} 	&{\scriptsize 11.814 }  & {\scriptsize 13.396}	 &{\scriptsize 2.025} &{\scriptsize 1.593} &{\scriptsize ---} \\
\hline
\end{tabular}
\end{center}
\vspace{-2ex}
\caption{\small  The average running time when machine floating type is used for all these methods.   }\label{tab:float type}
\end{table}

Table \ref{tab:experiment_result2}  lists the detailed results of Experiment 2.  
The results  show that the proposed method outperforms other four ones as well, and  is several orders of magnitude faster than the former two ones.  
By comparing the differences of these methods, we can easily see that the reason for the larger running time of the former two methods may well be that  both  methods  employ the C{\footnotesize GAL} library and   the D{\footnotesize CEL} structure{\small \footnote{We remark that computing the intersections of two polygons is unavoidable for any clipping algorithm, and it is a dominant step \cite{MarkdeBerg:computational,JamesDFoley:computer,GuntherGreiner:efficient}. Both the former two methods and the ``Standard'' method  adopt the plane sweep method to compute the intersections, their performance  differences however are so great. This reminds us that the gap may well be due to the usage of the C{\scriptsize GAL} library and   the D{\scriptsize CEL} structure (the former might be the major reason).   }}.  
Even so, it is noteworthy that comparing the former two methods with the latter three ones might be not very fair, as the latter three ones  use  \textit{floating point arithmetic} (similar to that in  \cite{GuntherGreiner:efficient,You-DongLiang:AnAnalysis,YongKuiLiu:AnAlgorithm,Patrick-GillesMaillot:aNew,AvrahamMargalit:Analgorithm,Francisco:aNew,YuPeng:aNew,AriRappoport:AnEfficient,MRivero:Boo,BalaRVatti:aGeneric}), whereas the  former two ones  use \textit{exact algebraic arithmetic}, which is more robust. 

To make a more fair comparison,  a simple remedy is to let the input data type be also \textit{machine floating point} for the  former two ones. In the next paragraph, we report our findings when all these methods use machine floating point type.

\subsubsection{Floating Point Number-type}
\label{subsec:floatingPointNumberType}


\begin{table}[t]
\begin{center}
\vspace{0ex}
\begin{tabular}{p{.08\textwidth} |p{.05\textwidth}  p{.05\textwidth} p{.05\textwidth} p{.05\textwidth} p{.05\textwidth} p{.05\textwidth} } \hline 
{{\scriptsize  Method}}&{{\scriptsize CGAL }}&{{\scriptsize Berberich}} & {{\scriptsize Naive }}& {{\scriptsize Standard }}&  {{\scriptsize RE2L }} \\\hline
{\scriptsize Time (sec.)} 		&{\scriptsize 0.07844 }  & {\scriptsize 0.08023} &{\scriptsize 0.00718} &{\scriptsize 0.004652} &{\scriptsize 0.002965} \\
{\scriptsize Imp. fac.} 	&{\scriptsize  26.4553 }  & {\scriptsize 27.058}	 &{\scriptsize  2.4215} &{\scriptsize 1.5689} &{\scriptsize ---} \\
\hline
\end{tabular}
\end{center}
\vspace{-2ex}
\caption{\small  The average running time when the two polygons satisfy the condition $l=m+n$.   }\label{tab:worst case}
\end{table}



Specifically, we use \texttt{double} type as the input data
type for all these methods, recall   Section \ref{subsec:datasets}. After this,   we  generate randomly a pair of circular-arc polygons and then attempt to call these methods. {Unfortunately, the former two methods fail with the
message 'precondition violation' for many inputs.}  To this step, we also attempt to generate many other (pairs of)  circular-arc polygons,  and to test them. {As a result, in most of cases   the runtime exceptions are also reported.}  We also realize that, for a few test data (i.e., circular-arc polygons generated randomly), all these methods can work correctly{\small \footnote{\footnotesize The former two methods are originally designed for the exact algebraic arithmetic, here we use the machine floating point as the input data type;  this could be the  reason why  most of cases cannot work correctly.  }}. For example, when the vertex coordinates of    two input circular-arc polygons (cf., Figure \ref{fig:usecase2})  are   \{(5,2), (\underline{5.125,1}), (6,0.5), (6,3), (4,5), (\underline{2.875,5.25}), (2,4.5), (1,4)\} and \{(2,4), (2,3), (\underline{1.25,2.75}), (1,2), (1,1), (4,1), (4,3),  (\underline{3.25,4})\}, all these methods can work correctly. We remember these two polygons and  run 100 times for each method (using these two polygons as the input). Table \ref{tab:float type}  depicts their average running time. This table shows us that the proposed
method also outperforms its simpler versions.  Again, the simpler versions of the proposed method
also outperforms the former two methods, which is similar to our previous findings (although the improvement factors decrease in terms of the former two methods). This   verifies our original claim in a more justified manner.

\subsubsection{The  Case  ``$l=m+n$'' }\label{subsec:theCASEequal}

Table \ref{tab:worst case} reports the results of  Experiment 4. Interestingly, the proposed method still outperforms  the `Standard' method{\small \footnote{We remark  that the results are similar when the input polygons with more edges are used, omitted for saving space.}} (note:  both  algorithms have the same time complexity in this case, recall Section \ref{sec:complexity}). This phenomenon reveals that (\romannumeral 1) two algorithms with the same time complexity might have different performance results in terms of execution time{\small \footnote{This argument could be another reason  why  the performance of the`Standard' method is significantly different from the former two methods. That is, it is possible that the C{\scriptsize GAL} library used in the former two methods  pays more attention to the robustness and stability, at the expense of the part of the performance.   The further explanation is beyond the topic of this paper.   }};  (\romannumeral 2) here other heuristics or optimization strategies (except the heuristic ``utilizing  related edges'') also bring us  the benefits; and (\romannumeral 3) the   performance gain of other optimization strategies  is larger than the cost consumed by the operator ``choosing related edges''.

\begin{figure}[t]
  \centering
      \includegraphics[width=29.15ex,height=21.23ex]{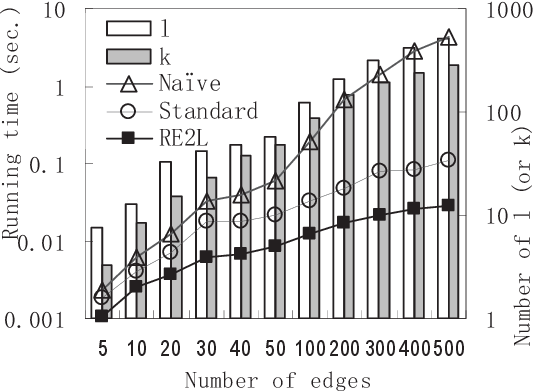} 
            \vspace{-1ex} 
 \caption{\small   The experimental results when we use larger data sets. } 
 \label{fig:addExp1}
\end{figure}

\subsubsection{Scalability}

Figure \ref{fig:addExp1} reports the results when we vary the number of edges (of polygons) from $5$ to $500$.  The columns denote the numbers of intersections and  related edges respectively, whereas the curves denote the running time. 
From this figure, we can  see that the R{\small E2L} has  better scalability compared to its competitors, as  the growth speed of the running time is slower than that of other two ones, when the number of edges   increases.   Compared to the `Standard' method, the better performance of our proposed method  is ascribed to those optimization strategies, and the poorer performance of the `Naive' method is due to that computing the intersections in such a way is (somewhat) inefficient. Especially,  this deficiency is more obvious when the number of edges of polygons is large.

\vspace{-2ex}

\section{Conclusions} \label{sec:conclusion} 

In this paper we investigated the problem of boolean operation on circular-arc polygons. By well considering the nature of the problem, concise  data structure and targeted algorithms were proposed. The proposed method runs in time $O(m+n+(l+k)\log l)$, using $O(m+n+k)$ space. 
We conducted extensive experiments demonstrating  the effectiveness and efficiency of the proposed method, and showed our techniques can be easily extended to compute boolean operation of other types of polygons. We conclude this paper with two research topics: (\romannumeral 1) 
As we know, the multiprocessor and multi-GPU systems are widely used nowadays; it could be interesting  to design efficient parallel algorithms for computing boolean operations of polygons. (\romannumeral 2) As we showed in the paper, our techniques can be easily extended to compute boolean operations of traditional polygons, and circular-ac polygons with holes and self-intersections;  it is still open whether our techniques can be extended to compute boolean operations on conic polygons or more general cases.